\documentclass[a4paper,cleveref,autoref,thm-restate,USenglish]{lipics-v2021}

\usepackage{soul}
\usepackage{url}

\usepackage{thm-restate}
\usepackage{enumitem}

\usepackage{graphicx}
\usepackage{comment}
\graphicspath{ {./images/} }

\long\def\todo#1 { {\bf TODO:} [{\color{blue} #1}] }

\usepackage{tikz}
\usepackage{graphicx}
\usepackage{amsthm}
\usepackage{amsmath,amsfonts}
\usepackage{amssymb}
\usepackage{ifthen}
\usepackage{comment}
\usetikzlibrary{calc}
\usepackage{multirow}

\usetikzlibrary{automata,positioning,fit,shapes.geometric,backgrounds}

\usepackage{csquotes}
\usepackage{booktabs}

\usepackage{mathtools}

\usepackage{algorithmicx,algorithm}
\usepackage[noend]{algpseudocode}

\newcommand{\calO}{\mathcal{O}}

\newcommand{\pseq}{\mathcal{S}}

\nolinenumbers

\title{Boosting Payment Channel Network Liquidity with Topology Optimization and Transaction Selection}

\titlerunning{Topology Optimization and Transaction Selection in PCNs}

\author{Krishnendu Chatterjee}{ISTA, Austria}{krishnendu.chatterjee@ist.ac.at}{}{European Research Council CoG 863818 (ForM-SMArt) and Austrian Science Fund 10.55776/COE12}
\author{Jan Matyáš Křišťan}{Faculty of Information Technology, Czech Technical University, Czechia}{kristja6@fit.cvut.cz}{}{Czech Science Foundation Grant no. 24-12046S}
\author{Stefan Schmid}{TU Berlin \& Weizenbaum Institute, Germany}{stefan.schmid@tu-berlin.de}{}{German Research Foundation (DFG) project ReNO (SPP 2378) from 2023-2027}
\author{Jakub Svoboda}{ISTA, Austria}{jakub.svoboda@ist.ac.at}{}{European Research Council CoG 863818 (ForM-SMArt) and Austrian Science Fund 10.55776/COE12}
\author{Michelle Yeo\footnote{Corresponding author}}{National University of Singapore, Singapore}{mxyeo@nus.edu.sg}{}{MOE-T2EP20122-0014 (Data-Driven Distributed Algorithms)}

\authorrunning{K. Chatterjee, J. Křišťan, S. Schmid, J. Svoboda and M. Yeo} 

\ccsdesc[500]{Networks~Network algorithms}
\ccsdesc[500]{Theory of computation~Graph algorithms analysis}
\ccsdesc[500]{Theory of computation~Approximation algorithms analysis}
\keywords{Blockchains, Cryptocurrencies, Payment Channel Networks, Throughput, Optimisation, Graph Algorithms, Approximation Algorithms} %TODO mandatory; please add comma-separated list of keywords

\EventEditors{Dariusz R. Kowalski}
\EventNoEds{1}
\EventLongTitle{39th International Symposium on Distributed Computing (DISC 2025)}
\EventShortTitle{DISC 2025}
\EventAcronym{DISC}
\EventYear{2025}
\EventDate{October 27--31, 2025}
\EventLocation{Berlin, Germany}
\EventLogo{}
\SeriesVolume{356}
\ArticleNo{4}
\begin{document}

\maketitle

\begin{abstract}
Payment channel networks (PCNs) are a promising technology that alleviates blockchain scalability by shifting the transaction load from the blockchain to the PCN. 
Nevertheless, the network topology has to be carefully designed to maximise the transaction throughput in PCNs. Additionally, users in PCNs also have to make optimal decisions on which transactions to forward and which to reject to prolong the lifetime of their channels.
In this work, we consider an input sequence of transactions over $p$ parties. Each transaction consists of a transaction size, source, and target, and can be either accepted or rejected (entailing a cost).
The goal is to design a PCN topology among the $p$ cooperating parties, along with the channel capacities, and then output a decision for each transaction in the sequence to minimise the cost of creating and augmenting channels,
as well as the cost of rejecting transactions.
Our main contribution is an $\calO(p)$ approximation algorithm for the problem with $p$ parties.
We further show that with some assumptions on the distribution of transactions, we can reduce the approximation ratio to $\calO(\sqrt{p})$.
We complement our theoretical analysis with an empirical study of our assumptions and approach in the context of the Lightning Network.
\end{abstract}

\begin{CCSXML}
<ccs2012>
   <concept>
       <concept_id>10003033.10003068</concept_id>
       <concept_desc>Networks~Network algorithms</concept_desc>
       <concept_significance>500</concept_significance>
       </concept>
   <concept>
       <concept_id>10003752.10003809.10003635</concept_id>
       <concept_desc>Theory of computation~Graph algorithms analysis</concept_desc>
       <concept_significance>500</concept_significance>
       </concept>
   <concept>
       <concept_id>10003752.10003809.10003636</concept_id>
       <concept_desc>Theory of computation~Approximation algorithms analysis</concept_desc>
       <concept_significance>500</concept_significance>
       </concept>
 </ccs2012>
\end{CCSXML}

\ccsdesc[500]{Networks~Network algorithms}
\ccsdesc[500]{Theory of computation~Graph algorithms analysis}
\ccsdesc[500]{Theory of computation~Approximation algorithms analysis}

\keywords{network algorithms, approximation algorithms, payment channel networks, cryptocurrencies}

\section{Introduction}
Blockchain scalability is one of the key bottlenecks in the mainstream adoption of cryptocurrencies~\cite{ChauhanMVM18,CromanDEGJKMSSS16,JainSG21}.
For instance, Bitcoin can only process an average of $7$ transactions per second which is paltry compared to Visa's $47,000$ \cite{poon2015lightning}. 
This makes it impractical as of yet for the widespread usage of cryptocurrencies in everyday situations. 

Some promising solutions to blockchain scalability are layer $2$ solutions like payment channel networks (PCNs)~\cite{decker2018eltoo,decker2015fast,mccorry2019pisa,poon2015lightning}, which allow users to bypass the blockchain and transact directly with each other.
To do so, two users have to first open a payment channel between themselves in an initial funding transaction on the blockchain, thereby locking some funds only to be used in the channel. 
Thereafter, the two users transact with each other simply by updating the distribution of funds on each end of the channel to reflect the payment amount. 
As compared to the tedious process of waiting for consensus on the blockchain, these payments can be finalised almost instantaneously since they only involve exchanging signatures between the two parties. 
As such, with only a constant number of blockchain transactions, any two users can make an unbounded number of off-chain transactions in the payment channel, consequently increasing the overall transaction throughput of the blockchain. 

A PCN is a network of payment channels and users, where two users that are not directly connected to each other with a payment channel can still transact with each other in a multi-hop fashion over a path of payment channels. 
To incentivize intermediary nodes on payment paths to forward payments, PCNs allow these intermediary nodes to charge a small transaction forwarding fee that is typically linear in the transaction amount.
%Intermediary nodes on a payment path typically charge a transaction fee to forward payments, as forwarding payments could lead to an unfavourable shift in the distribution of funds in their payment channels.
Two of the most notable examples of PCNs are Bitcoin's Lightning Network~\cite{poon2015lightning} and Ethereum's Raiden~\cite{raiden}. 

Apart from joining a PCN in order to transact efficiently with other users, a secondary consideration for users in PCNs is to whom they should establish channels and at which capacities, as well as which transactions then to select for forwarding.
Aside from the opportunity cost of rejecting to forward transactions and thus missing out on the revenue from transaction fees, users in PCNs also incur a cost whenever channels are created (the on-chain transaction fee) as well as when additional capacity is injected into the channel.
As such, although it is tempting to accept to forward all transactions and thus profit from the transaction fees, doing so would quickly deplete the nodes' capacity on their incident channels. 
%In order to maintain liveness, they would then have to go back to the blockchain to close and reopen a new channel with more funds, which could incur greater cost. 
A crucial optimisation problem that intermediary nodes in transaction paths would therefore have to consider is deciding which transactions to accept and which to reject, given a minimal injection of capacity into their channel. 

In this work, we study the problem of both optimal topology design and transaction selection over a PCN. 
The input to the problem is a sequence of transactions and a set of $p$ coordinating parties.
Each transaction in the sequence is of a certain amount, and is associated with a source and target from the set of parties.
The problem then proceeds in two stages: in the first stage, the goal is for the $p$ parties to decide jointly on a PCN topology over themselves. 
Each created channel in this stage incurs both a fixed channel creation cost as well as a capacity cost that depends on the amount of funds injected into the channel.
The second stage involves processing the transactions in the transaction sequence, either accepting or rejecting them (and in so doing incurring a rejection fee) along the created paths in the network. 
The goal is to create a network and inject enough capacity into each channel such as to maximise the number of transactions accepted while incurring the least amount of cost (see the cost model in~\Cref{sec:model} for more details). 

\subsection{Related work}
To the best of our knowledge, our work is the first to consider the problem of designing an optimal PCN topology for optimising the problem of transaction selection in PCNs.
Our work is closely related to the work of~\cite{SchmidSY23} which considers a similar problem over a single payment channel and shows that the problem of deciding which transactions to accept or reject over a single channel is already NP-hard.
The main open question of~\cite{SchmidSY23} is how to extend their proposed algorithm from a single channel to an entire network.
Our work addresses this problem by proposing a solution for network design as well as algorithms for optimally accepting or rejecting transactions along paths in a network.  
Our work is also related to the works of~\cite{AvarikiotiB0W19,online} which consider a similar problem of transaction selection but in the online setting and also limited to a single channel.
Additionally, the works of~\cite{lntop,GuasoniHS24} also study minimal-cost PCN topologies, but do not consider the complementary problem of transaction selection.

Another related line of research is network creation games on cryptocurrency networks~\cite{AvarikiotiH0W20ride,lcg,DBLP:conf/fc/ErsoyRE20}. These works study how a node should connect to an existing cryptocurrency network in a stable and optimal way, optimising the utility of the single node in a unilateral fashion. In contrast, our work proposes a method of creating a network among any number of nodes such as to \emph{jointly minimise} the cost (over all nodes) of network creation and processing transactions.

A crucial assumption in our model is coordination between the interested parties in creating the network during the network creation stage. 
We stress that this is not an impractical assumption and several works have made similar assumptions~\cite{AvarikiotiH0W20ride,lcg,online} where optimal behaviour is analysed in the setting where some or all involved parties cooperate and create certain network topologies together.

Further upfield but also related are works analysing the consensus number of a cryptocurrency~\cite{GuerraouiKMPS19,herlihy91}. 
The transactions studied in our work, and in general the sequential nature of layer 2 transactions on PCNs, have consensus number 1, implying that expensive consensus protocols are unnecessary in such settings. 
Additionally, our work is also related to the classic works on theory of contention management in transactional memory~\cite{AttiyaEST10,poly05,GuerraouiHP06}.

More generally, our work is related to optimising flows and throughput in typical capacitated communication networks~\cite{chekuri2004all,raghavan1985provably}.
However, we stress a crucial difference between our problem and problem over traditional communication networks:
in traditional communication networks, the capacity is usually independent in the two directions of the channel~\cite{Gupta2000TheCO}. 
In our case, the amount of transactions $u$ sends to $v$ in a channel $(u,v)$ directly affects $v$'s capability to send transactions, as each transaction $u$ sends to $v$ increases $v$'s capacity on $(u,v)$.

\subsection{Main challenges and our contributions}
We generalise the model of processing transactions over a single channel proposed in~\cite{SchmidSY23} to an entire network. 
The first challenge in shifting from optimising processing transactions from a single channel to a network is deciding on the optimal network topology among the cooperating parties.
Here, we show in~\Cref{thm:star} in~\Cref{sec:star} that the star graph of size $p$ is a $2$-approximation of the problem of creating an optimal PCN over $p$ parties.

Once the optimal topology is decided, the next challenge would be to decide on which transactions to accept to forward, given that transactions have to be accepted along the entire payment path. 
%Note that this restriction that transactions have to be accepted or rejected along entire paths already implies that more capacity has to be injected into channels as compared to the problem over a single channel.
In this setting, we leverage the star topology to generalise the existing single channel algorithm in~\cite{SchmidSY23} to get an algorithm for transaction selection optimisation in the star.
Our main algorithm in this setting is described in \Cref{sec:alggeneral} and we prove that it is an $\calO(p)$ approximation of the weighted transaction selection problem in general graphs.

We describe another algorithm in~\Cref{sec:algcluster} that has an improved approximation ratio of $\calO(\sqrt{p})$ (see \Cref{thm:main}) for the transaction selection problem in general graphs. 
To achieve this approximation ratio, we adopt an additional realistic and reasonable assumption on the distribution of transactions in the transaction sequence. Namely, we assume that transactions are clustered, following a stochastic block model.

Finally, we augment our theoretical analysis with a case study of the Lightning Network in~\Cref{sec:eval}. Using simulations, we find that the network statistics correspond to and thus justify our transaction distribution assumptions. 
%We conclude our work by outlining future directions in~\Cref{sec:conclusion}.

\section{Model}\label{sec:model}

\smallskip{\em Payment channels.}
A unique aspect of our problem is that we optimise the selection of weighted transactions over payment channels. Payment channels can be formally modeled as edges that satisfy the following three properties:

\begin{enumerate}[leftmargin=0.5cm]
    \item The capacity or weight of a channel $(u,v)$ is the sum of the initial capacities $b_u$ and $b_v$ injected into the channel by users (or nodes) $u$ and $v$. The capacities $b_u$ and $b_v$ correspond to the initial amount of funds that users $u$ and $v$ choose to inject into the channel during channel creation.
    \item Given a capacitated channel $(u,v)$ with weight $w$, the weight or capacity can be arbitrarily split between both ends of the channel depending on the number and weight of transactions processed by $u$ and $v$. 
    That is, $b_u$ and $b_v$ can be arbitrary as long as $b_u + b_v = w$ and $b_u, b_v \geq 0$. 
    For instance, given a payment channel $(u,v)$ with an initial capacity split of $b_u$ and $b_v$, if $u$ sends a payment of amount $x$ to $v$, the capacity on $u$'s end of the channel drops to $b_u-x$ and the capacity of $v$'s end increases to $b_v+x$ after processing the transaction.
    See~\Cref{fig:example} for more examples of how the capacities on each end of a payment channel can vary in the course of processing transactions.
    \item The total capacity of a payment channel is invariant throughout the lifetime of the channel. 
    That is, it is impossible for nodes to add to or remove any part of the capacity in the channel. 
    In particular, if a node is incident to more than one channel in the network, the node cannot transfer part of its capacity in one channel to ``top up'' the capacity in the other one. 
\end{enumerate}

\begin{figure}[htb!]
    \centering
    \includegraphics[scale=0.7]{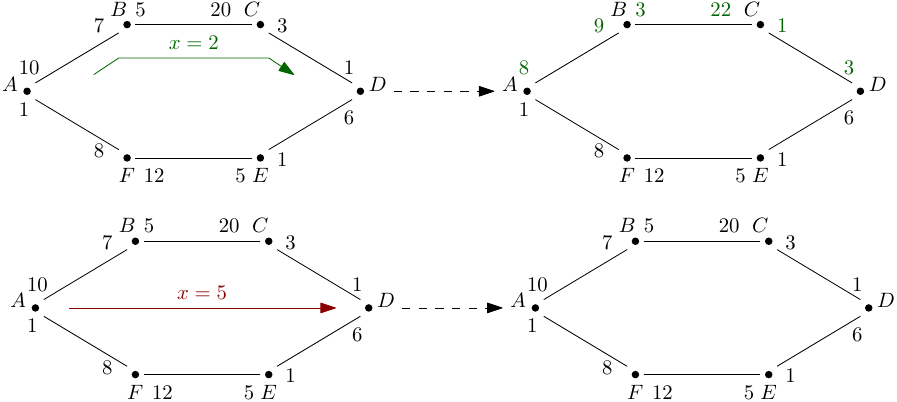}
    \caption{The diagram on the top shows the outcome of processing a transaction of size $2$ going from $A$ to $D$ along the path $(A,B,C,D)$ in the graph. 
    The diagram on the bottom shows the outcome of processing a transaction of size $5$ from $A$ to $D$. The transaction has to be rejected as the size of the transaction is larger than the sum of the capacities of $C$ in the channel $(C,D)$ and either $A$ or $E$ in the channel $(A,F)$ or $(E,D)$ respectively.}
    \label{fig:example}
\end{figure}

\smallskip{\em Payment channel network.}
We model a PCN as a weighted, undirected graph $G=(V,E)$ where the vertices or nodes of $G$ are users in the PCN and the edges of $G$ are payment channels between nodes. 
%Let $G=(V,E)$ be a weighted, undirected graph. 
%We call a graph a rechargeable if all its links are rechargeable. 
We denote the capacity of a channel $(u,v) \in E$ as $w_{u,v}$, and the capacity of nodes $u$ and $v$ on the channel $(u,v)$ as $b_{u,v}$ and $b_{v,u}$ respectively.

\smallskip{\em Transaction sequence.}
Let $V$ denote a set of nodes (parties). 
We denote a finite ordered sequence of transactions over $V$ by $X_n = ((s_1, t_1, x_1), \dots, (s_n, t_n, x_n))$, where $s_i, t_i \in V$ represents the sender and target of the $i$th transaction respectively and $x_i \in \mathbb{R}^+$ represents the weight or size of the $i$th transaction.
%we further specify the transaction sequence to be $X_n = ((\pi_1, x_1), \dots, (\pis_n, x_n))$, where each $\pi_i$ is a path in $G$ that specifies the path in the graph to forward the transaction from sender to target, henceforth called ``transaction path".
%For an edge $e = (u,v)$ in a given transaction path $\pi$, we also call the first node $u$ in the edge the forwarding node.
%For the special case where nodes in $G$ are connected by unique paths (e.g., if $G$ is a star), we omit the usage of paths in specification of the sequence and simply use the sender and target specification.

\smallskip{\em Processing transactions in PCNs.}
Let $G=(V,E)$ be a connected PCN with some capacity over all of its channels. 
Let $X_n$ be a sequence of transactions over $V$.
Consider the $i$th element of the sequence $(s_i, t_i, x_i)$ and suppose the transaction travels over some path $\pi_i = (s_i, \dots, t_i)$ from the $i$th sender to the $i$th target.
Let us call the first node along each channel in the path $\pi_i$ the \emph{forwarding node}.
Each forwarding node in the path $\pi_i$ can choose to do the following to the transaction:

\begin{itemize}[leftmargin=0.5cm]
    \item \textbf{Accept transaction.} For a forwarding node $u$ in a channel $(u,v) \in \pi_i$, node $u$ can accept to forward the transaction if their capacity in $(u,v)$ is at least the size of the transaction. 
    %The result of doing so decreases their capacity by the transaction weight and increases the capacity of $v$ by the transaction weight.
    \item \textbf{Reject transaction.} Node $u$ can also reject the transaction. This could happen if $u$'s capacity is insufficient, or if accepting the transaction would incur a larger cost in the future. For a transaction of weight $x$, the cost of rejecting the transaction is $f \cdot x + m$ where $f,m \in \mathbb{R}^+$.
    Apart from incurring some cost, rejecting a transaction does not alter the capacity distribution on both ends of the channel (see the diagram on the bottom of~\Cref{fig:example} for an example of rejecting a transaction).
\end{itemize}
We note that node $v$ does not need to take any action nor incur any cost when transactions are going in the direction of $u$ to $v$ along a channel $(u,v)$.
Finally, we stress that decisions on each channel in a transaction path need to be the same.
%That is, for a transaction travelling along a path $\pi_i = (e^i{_1}, \dots, e^i{_j})$, if the decision on a single channel say $e^i{_1}$ is to reject a transaction, all the other channels $e^i{_2}, \dots, e^i{_j}$ in the path would also have to reject the transaction. The total cost of rejecting the transaction on this path would then be $j (f\cdot x+m)$.

\smallskip{\em Problem definition.}
Suppose we have $p$ coordinating parties that want to come together to create a PCN that allows the processing of transactions between themselves. 
The input to the problem is a transaction sequence over the $p$ coordinating parties. 
The problem then proceeds in $2$ stages. 
The first stage of the problem is the \emph{network creation stage}, where the $p$ parties have to firstly decide on a PCN topology over themselves.
Once the PCN has been created, the second stage of the problem is \emph{transaction selection}, where the network topology created in stage $1$ determines the path the transactions travel over.
In this stage of the problem, the parties need to decide whether to accept or reject each transaction in the sequence.

\smallskip{\em Hardness.}
We know from~\cite{SchmidSY23} that the transaction selection problem over a single channel is NP-hard.
Our problem includes $p>2$ parties and we also need to find the optimal network topology over these $p$ parties (for two parties, it is only a channel).
This means that our problem is at least as hard as the problem on a single channel, and hence NP-hard as well.

\smallskip{\em Objective and costs.}
We consider three types of costs faced by the $p$ parties: 
\begin{enumerate}[leftmargin=0.5cm]
    \item \textbf{Channel creation costs.} 
    %These are auxiliary costs involved with creating a channel between any $2$ parties during the network creation stage. 
    We assume creating each channel incurs a fixed auxiliary cost of $k>0$. 
    \item \textbf{Capacity costs.}
    For each channel $(u,v)$ created between parties $u$ and $v$, the capacity cost of $(u,v)$ is the sum of the capacities $b_{u,v},b_{v,u}$ injected by $u,v$ into the channel. 
    %These are costs incurred at the start of transaction selection.
    \item \textbf{Rejection costs.}
    These are costs incurred during transaction selection when a party rejects a transaction. 
\end{enumerate}
The \emph{total cost} incurred by the parties is the sum (over all channels in the PCN and transaction decisions) of channel creation, capacity, and rejection costs. 
The objective of the $p$ parties is to create a PCN over themselves and output the decisions for each transaction in the sequence such that the total cost is minimised. 

\section{Accepting all transactions}

We begin our investigation of transaction selection over general graphs by first considering the problem a restricted setting where all transactions in the input transaction sequence need to be accepted.
In this setting, the total cost consists only of channel creation and capacity costs.
We also assume here that we know the optimal PCN topology among the $p$ parties, which we denote by $G = (V,E)$. 

\subsection{Linear program for a set of channels}

We present a linear program (LP) that, for a given set of channels, computes the optimal capacity that the parties need to lock into the channels of the network in order to accept all transactions in the sequence.
Indeed, if the channel creation cost $k$ is small or if we want to know a lower bound for our problem, we can assume the optimal graph $G$ is complete.
For a node $v \in G$, let $Ne(v)$ denote its neighbours.
Let $S_v$ denote the \emph{total balance} of node $v$ in the graph, which is simply the sum of its balances on all its incident channels. 
For a channel $(u,v) \in E$, we denote the balance of node $v$ (resp. $u$) after processing the $i$th transaction in the sequence by $b^i_{v,u}$ (resp. $b^i_{u,v}$), and we also denote the total balance of $v$ after processing the transaction by $S^i_v$.
Let $AC(\mathcal{S}, G)$ denote the optimal capacity cost for the graph $G$ where every transaction in $X_t$ is accepted. That is, the total capacity injected into the channels of the graph $G$. 

%If the cost of creating an channel is low or we want to know the lower bound, we can suppose the graph is complete.

Before describing the LP, we define a \emph{difference function} $\delta_v(s,t)$ between a source $s$ and target $t$ in a transaction path:
\[
    \delta_v(s,t) = 
    \begin{cases}
        1 \qquad  & v = t \text{ and } v \neq s\\
        -1 \qquad & v = s \text{ and } v \neq t\\
        0  \qquad & \text{otherwise}
    \end{cases}
\]
This function will be used in the LP to maintain the invariant that processing a transaction can only affect the total balances of the source and target nodes in the transaction path, but the total balance of every other node in the graph must not change.
%That is, after processing transaction $i$ of size $x_i$, going from source $s_i$ to target $t_i$, the total balance of $s_i$ needs to decrease by $x_i$ and the total balance of $t_i$ needs to increase by $x_i$, but the total balance of every other node in the graph must not change.

Given $G$ and a transaction sequence $\pseq := ((\pi_1,x_1), \dots, (\pi_n,x_n))$ over $G$, we construct the following LP.

\begin{align*}
  \qquad\qquad \text{minimise } \sum_{(v,u) \in E} w_{v,u}\\
  \forall v,u,i :& \qquad w_{v,u} = b_{v,u}^i + b_{u,v}^i\\
  \forall v,i :& \qquad S_v^i = \sum_{u \in Ne(v)} b_{v,u}\\
  \forall v, i, (s_i, t_i, x_i) : & \qquad S_{v}^i-\delta_v(s_i,t_i)x_i = S_{v}^{i-1}
\end{align*}

%The LP minimises the sum of the capacities over all channels of the graph.
For the sequence of transactions $\pseq$ and graph $G$, we denote the solution of the LP with the respective parameters by $LP(\pseq, G)$.
The following lemma (with proof in~\Cref{app:prooflp}) proves that the solution of the LP is a lower bound on the capacity cost of the problem in the setting where all transactions are to be accepted. 

\begin{restatable}[]{lemma}{lp}
\label{lemma:lp}
  Given the sequence of transactions $\pseq$,  we have $LP(\pseq, K_p) \leq AC(\pseq,G)$.
\end{restatable}

%\begin{lemma}\label{lemma:lp}
%  Given the sequence of transactions $\pseq$,  we have $LP(\pseq, K_p) \leq AC(X_t,G)$.
%\end{lemma}

\subsection{Star as a $2$-approximation}\label{sec:star}

Although the capacity cost for the setting where all transactions are accepted is lowest in the complete graph $K_p$, the complete graph has $\Theta(p^2)$ edges. Thus, depending on the auxiliary channel creation cost $k$, the cost of creating $K_p$ as well as injecting sufficient capacity into all channels of $K_p$ might be large.

As a useful heuristic, we observe from the proof of~\Cref{lemma:lp} that the capacity cost is related to the diameter of the graph: the smaller the diameter, the lower the cost.
%On the other hand, we observe that every transaction sequence that involves $p$ parties, i.e., each party is involved in at least one element of the transaction sequence, would require creating a graph of at least $p-1$ channels during the network creation stage.
%(if it can be disconnected, we solve every part by itself).
Thus, a useful class of graphs to look at for the network creation stage of the problem are graphs that have both a small diameter and a small number of channels. 
An example of a graph that simultaneously satisfies both of these requirements is a star.
In a star, one central node is connected by a channel to all other nodes.
Thus, the star has a small diameter of $2$, and $p-1$ channels, which is minimal for a connected graph over $p$ nodes.
We denote the star on $p$ nodes as $K_{1,p-1}$.

We first observe that for the star, it is easy to compute a lower bound on the amount of capacity required in the setting where all transactions are accepted. 
This is due to the fact that between any two nodes $u,v$ in a star, there is only one unique path channeling $u$ and $v$, rendering transactions unsplittable.
The capacity on the channel between any node $v$ and the centre of the star $c$, $w_{v,c}$, has to therefore be larger than $\max_{j,k \le n} \quad |\sum_{s_i,t_i,x_i | k \le i \le j} \delta_v(s_i,t_i) \cdot x_i|$, as this quantity captures the largest difference in the balance between both users of the channel.

We now show in~\Cref{thm:star} that we can upper bound the optimal capacity cost on the star. The proof of~\Cref{thm:star} is in~\Cref{app:proofstar}.

\begin{restatable}[]{theorem}{star}
\label{thm:star}
  For every graph $G$ over $p$ nodes, and any graph $G'$ with diameter $d$, the following inequality holds
  \[
    AC(\pseq, G') \le d \cdot AC(\pseq, G) \,.
  \]
\end{restatable}

\begin{corollary}\label{cor:star}
  For every graph $G$ over $p$ nodes, and graph $K_{1,p-1}$, the following inequality holds
  \[
    AC(\pseq, K_{1,p-1}) \le 2 \cdot AC(\pseq, G) \,.
  \]
\end{corollary}

\begin{comment}
\begin{remark}
    Since the diameter of the star is $2$,~\Cref{thm:star} shows that the optimal capacity cost on the star is at most twice as high as the optimal capacity cost on any graph while the cost of creating channels is minimal (see~\Cref{cor:star}).
    We also note that the result in~\Cref{thm:star} holds also in the setting where some transactions are rejected. 
    In this setting, given the solution (decisions to accept and reject transactions) on the optimal graph for input transaction sequence $X_t$, we look only at the subsequence $X'_t$ formed by transactions that are accepted, and then apply~\Cref{thm:star}.
\end{remark}
\end{comment}

With these insights, we restrict the network created among the $p$ parties during the network creation stage of the weighted transaction problem to a star in~\Cref{sec:alggeneral}, and a star of stars (of diameter $4$, see~\Cref{fig:doublestar} for an example) in~\Cref{sec:algcluster}.

\section{A general setting approximation algorithm}\label{sec:alggeneral}
We now present an approximation algorithm for the general setting where transactions can be either accepted or rejected.
We first reiterate that from \Cref{thm:star} we can restrict ourselves to the star graph at the network creation stage of the problem, and in doing so incur twice as much capacity cost as the optimal capacity cost.
Given that the channel creation cost for the star topology is fixed at $k(p-1)$, the focus of our algorithm in this section is to come up with capacities that should be injected on both ends of each channel as well as a sequence of transaction decisions so as to minimise the capacity and rejection costs.

%The important part of any algorithm is to balance the rejection and liquidity costs.
%These costs are orthogonal in nature: on one hand we can reject everything and incur large rejection costs with $0$ liquidity costs.
%On the other hand, we need enough liquidity to ensure that every transaction is accepted in order to incur $0$ rejection costs.

%We first describe in~\Cref{sec:lpstar} a LP formulation which is a relaxation of our general problem setting, where we allow the LP to accept a portion of each transaction in the sequence. We then present a deterministic algorithm in~\Cref{sec:detalg} for this setting that closely mimics the solution of the LP to process transactions between every pair of leaf nodes in the star.

\subsection{LP for a star}\label{sec:lpstar}
We describe a LP that minimises the capacity cost as well as rejection cost over transactions going between any two leaf nodes in the star. This LP is a relaxation of the general problem, where we allow some transactions to be accepted fractionally.
That is, for a transaction of size $x_i$, the LP does not have to accept the full transaction amount, but can accept some portion $y_i$ such that $0 \leq y_i \leq x_i$. Let us call transactions for which $y_i=x_i$ in the LP solution \emph{fully accepted} transactions.

Let $c$ denote the central node of the star. 
For a leaf node $v$ and $i \in \mathbb{N}$, we use
$C_{v,i}$ and $C_{v,i}'$ to denote the capacity at the channel $(v,c)$ after processing the $i$-th transaction
on the side of $v$ and $c$ respectively.
Let $C_v$ denote the total capacity of the channel $(v,c)$.
We can now formulate the LP as follows:

\begin{align}\label{lp:star}
  \qquad\qquad \text{minimise }&\sum_{u} C_u +  \sum_{i} f(x_i-y_i) + \frac{m(x_i-y_i)}{x_i}\\
  \forall i :& 0 \le y_i \le x_i \notag\\
  \forall (v,u,x_i) :&  C_v, C_u \ge x_i\notag\\
                     &  C_{v,i} = C_{v,i-1} - y_i\notag\\
                     &  C_{v,i}' = C_{v,i-1}' + y_i\notag\\
                     &  C_{u,i}' = C_{u,i-1}' - y_i\notag\\
                     &  C_{u,i} = C_{u,i-1} + y_i\notag
\end{align}

%Note that we assume that $C_v \ge x_i$, which means that there are no transactions that are too large.
%This is a reasonable assumption in practise as large transactions are typically split into multiple smaller transactions along the same route~\cite{SivaramanVRNYMF20}. 
%Let us call transactions for which $y_i=x_i$ in the LP solution \emph{fully accepted} transactions.
\begin{lemma}
    The total cost of the star LP solution is a lower bound on the total cost of the optimal solution.
\end{lemma}
\begin{proof}
    The proof follows from the observation that the optimal integer solution is also a feasible solution for the LP.
\end{proof}

\begin{remark}\label{rem:fully}
The solution of the LP assumes that the capacity cost and the transaction rejection cost are equally large, i.e., we cannot aggressively optimise for one objective while neglecting the other.
An implication of this which motivates the design of our algorithms in the following sections is that our algorithms should infrequently reject fully accepted transactions, as these are transactions that do not incur any rejection cost and also change the distribution of the capacity along channels in a more optimal way.
\end{remark}

\subsection{An $\calO(p)$ approximation algorithm}\label{sec:detalg}

\smallskip{\em Algorithm for a single channel.} We summarize an algorithm from~\cite{SchmidSY23} that solves the problem for a single channel, which we will call the \emph{channel algorithm}.
The channel algorithm first approximates the capacity to be injected into the channel.
Next, the channel algorithm solves a LP (similar to \Cref{lp:star}).
Given the solution of the LP, we denote by $C_u$ and $C_v$ the minimal capacity needed on each end of a channel $(u,v)$ to mimic the (fractional) transport of transactions as dictated by the LP.
In addition to $C_u$ and $C_v$, the channel algorithm allocates additional capacities on each end of the channel, which we call \emph{reserves}, and denote them by $R_u$ and $R_v$. 
These reserves are mainly used to ensure that all transactions that are fully accepted by the LP are also accepted by the channel algorithm (a consequence of~\Cref{rem:fully}).
%, and secondly, ensure that as many transactions $x_i$ such that the LP output $\frac{y_i}{x_i}$ is sufficiently large are also accepted by the channel algorithm. 

%In the channel algorithm, if the $i$-th transaction from $u$ to $v$ is fractionally accepted by the LP and also accepted by the channel algorithm, the fraction $y_i$ that is accepted by the LP is moved from $C_u$ to $C_v$, and the additional $x_i-y_i$ capacity is transported from $R_u$ to $R_v$.
%If the transaction is rejected, the transport of capacity $y_i$ from $u$ to $v$ is simulated.
%In other words, $y_i$ capacity is still moved from $C_u$ to $C_v$. 
%However, to ensure that the \emph{total capacity} of $u$ does not change when rejecting a transaction, the algorithm moves $y_i$ capacity from $R_v$ to $R_u$.
%Finally, we stress that the channel algorithm ensures that the reserves are always non-negative.

\smallskip{\em Extending the channel algorithm to a star.}
We describe our algorithm that extends the channel algorithm over $2$ nodes to a star $K_{1,p-1}$ with central node $c$ in~\Cref{alg:modif}.
The key idea behind our algorithm is to first use the LP as described in~\Cref{lp:star} to fix initial capacities on each end of the channel for every channel in the network (for instance the capacity $C_v$ for $v$ in $(v,c)$ in~\Cref{fig:channeltostar}).
We also provide each leaf with $p-2$ additional reserves (to process transactions for each other leaf). 
For each leaf pair $u$ and $v$, let $M$ denote the smallest change in the capacity as specified by the output of the LP after all interactions between $u$ and $v$.
The additional reserve $R_u^v$ of $u$ for processing transactions of $v$ would then be of size $\frac{\sqrt{3}}{2}M$ (similarly for $v$).
We leave details of~\Cref{alg:modif} to~\Cref{app:channelalg}.
%Given that our algorithm always maintains at least the same amount of funds as prescribed by the LP on each end of every channel, every transaction that is fully accepted by the LP is also accepted automatically by~\Cref{alg:modif}.

\begin{figure}[htb!]
    \centering
\includegraphics[scale=0.7]{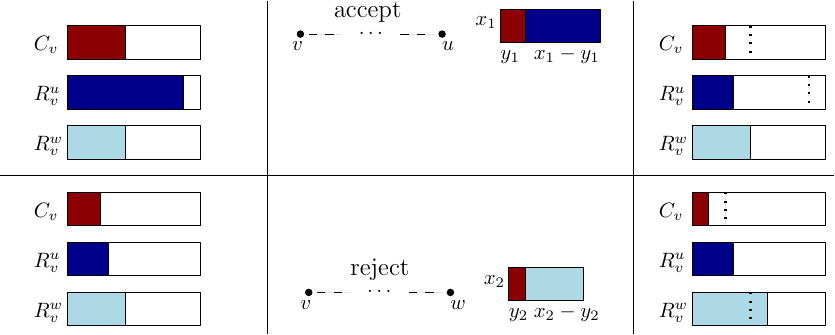}
    \caption{Movement of capacities and reserves from the perspective of node $v$ in a star with centre node $c$ and $2$ other leaf nodes $u$ and $w$. $C_v$ represents the capacity as specified by the output of the LP. $R_v^u$ and $R_v^w$ denote the reserves stored by $v$ to process transactions from $u$ and $w$ respectively. The rightmost column represents the capacities and reserves of $v$ after processing the transaction.}
    \label{fig:channeltostar}
\end{figure}

The next theorem shows that~\Cref{alg:modif} is an $\calO(p)$ approximation of the optimal cost on a star. The proof of~\Cref{thm:pairwise} is in~\Cref{app:proofpairwise}

\begin{algorithm}[t]
  \begin{algorithmic}[1]
    \Require nodes $u$ and $v$, transaction sequence $X_t$, solution of $LP$ indexed by $i$:$C_{u},C_{u}',C_{v},C_{v}'$.
    \Ensure decisions to accept or reject for transactions between $u$ and $v$
    \State $M \gets \min(C_{u,0} + C_{u,0}', C_{v,0} + C_{v,0}')$ \label{algline:reserves1} 
    \State $R_{u} = \frac{\sqrt{3}}{2}M$, $R_{v} = \frac{\sqrt{3}}{2}M$ \label{algline:reserves2}
    \State $H \gets \Call{CreateHeap}{ }$
    \For{$i \in [t]$}
     \If {$x_i$ is transaction from $u$ to $v$}
      \If{$R_u - (x_i - y_i) \ge \frac{\sqrt{3}-1}{2}M$} \label{algline:balstart}
        \State $x_i \gets \textbf{Accept}$
        \State $R_u, R_v = R_u - (x_i-y_i), R_v + (x_i-y_i)$
        %\State $R_v = R_v + (x_i-y_i)$ 
      \ElsIf{$\frac{y_i}{x_i} < \frac{\sqrt{3}}{\sqrt{3}+1}$} \label{algline:rej}
        \State $x_i \gets \textbf{Reject}$
        \State $R_u, R_v = R_u + y_i, R_v - y_i$
        %\State $R_v = R_v - y_i$ 
      \Else
        \State $x_i \gets \textbf{Accept}$
        \State $R_u, R_v = R_u - (x_i-y_i), R_v + (x_i-y_i)$
        %\State $R_v = R_v + (x_i-y_i)$
        \State $\Call{AddToHeap}{H,x_i}$
        \If{$R_u < 0$}
            \While{$R_u < \frac{\sqrt{3}-1}{2}M$} \label{algline:forward}
                \State $x_j \gets \Call{PopHeap}{H}$
                \State $x_j \gets \textbf{Reject}$
                \State $R_u, R_v = R_u + x_j, R_v - x_j$
                %\State $R_v = R_v - x_j$
            \EndWhile
        \EndIf
      \EndIf     
      \EndIf
       \If{$R_u \ge \frac{\sqrt{3}-1}{2}M$}
            \State $\Call{EmptyHeap}{H}$
       \EndIf
    \EndFor
  \end{algorithmic}
  \caption{Modified channel algorithm}\label{alg:modif}
\end{algorithm}

\begin{restatable}[]{theorem}{pairwise}
\label{thm:pairwise}
  Given the transaction sequence $X_t$, let $\mathcal{C}$ be the optimal capacity cost and $\mathcal{R}$ the optimal rejection cost on a star. Then~\Cref{alg:modif} run for every pair of nodes
  incurs a rejection cost of at most $(\sqrt{3} + 1)\mathcal{R}$ and capacity cost of at most
  $(1 + (p-1)\cdot \sqrt{3})\mathcal{C}$.
\end{restatable}

In the next section, we show that if we impose some reasonable assumptions on the distribution of the input transaction sequence in our model, we can design an algorithm with a significantly lower approximation ratio of $\calO(\sqrt{p})$.

\section{An improved algorithm using the stochastic block model}\label{sec:algcluster}

%We first describe our notion of clustering in~\Cref{sec:cluster}, and show how this assumption influences the topology of the network created during the network creation stage of the problem. We then describe in~\Cref{sec:clusteralg} our algorithm which uses the clustering assumptions to give a $\calO(\sqrt{p})$ approximation of the optimal cost for the problem.
%We include in~\Cref{sec:discussion} a discussion on relaxations and further optimisations of our theoretical results, which may be of independent interest.

\smallskip{\em Stochastic block model.}
In reality, users typically interact with each other in a structured manner, resulting in spatial and temporal locality. For example, if two nodes that previously transacted with each other are more likely to transact again in the near future.  
Our model for how nodes interact and send transactions to each other is the stochastic block model~\cite{HOLLAND1983109}, in particular, the planted partition model~\cite{CondonK01}.
The stochastic block model is a generative model to create random graphs where nodes are partitioned into clusters (blocks).
See~\Cref{app:sbm} for more details on the stochastic block model.

\smallskip{\em Clustering.}
We now define the way nodes are partitioned into clusters given an input sequence of transactions. 
Given an input transaction sequence over a network, we say that the nodes are \emph{$(m,k,t)$-clustered} if there exists a partition of nodes
to $m$ clusters $C_1, C_2, \dots, C_m$ where each cluster has size at most $k$ and
the clusters satisfy the following \emph{clustering conditions}:
\begin{enumerate}[leftmargin=0.5cm]
    \item The volume (sum of sizes of transactions) of transactions inside the cluster is $t$-times
    the volume of transactions that go outside the cluster. That is,
    \[
        \sum_{(s_j,t_j,x_j) | s_j,t_j \in C_i  } x_j \ge t \cdot \sum_{\substack{(s_j,t_j,x_j)| s_j \in C_i \text{ and } t_j \not\in C_i \\ \text{or } s_j \not\in C_i \text{ and } t_j \in C_i}} x_j
    \]
    \item For  every $S \subset C_i$ such that $|S| \le \frac{C_i}{2}$,
    the volume of transactions between $S$ and $C_i \setminus S$
    is at least $\frac{|S|}{4|C_i|}$ the volume in the cluster $C_i$. That is,
    \[
        \frac{|S|}{4|C_i|} \sum_{x_j \text{ inside } C_i} x_j \le \sum_{x_j \text{ between } C_i\setminus S \text{ and } S} x_j
    \]
    \item For  every $S \subset C_i$, the ratio of volumes of transactions from $S$ and $C_i \setminus S$
    is at least $\frac{1}{3}$. That is,
    \[
        \frac{\sum_{x_j: C_i \setminus S \rightarrow  S} x_j}{\sum_{x_j: S \rightarrow C_i \setminus S} x_j} \ge \frac{1}{3}.
    \]
\end{enumerate}
These conditions need to be satisfied for every cluster $C_i$ and any contiguous subsequence of transactions involving $C_i$ (i.e., source or target in $C_i$) with volume at least $C_i \cdot x_{max}$,
where $x_{max}$ is the weight of the biggest transaction in the input transaction sequence.
In the rest of this section, we assume that given an input sequence of transactions, the nodes can be $(\sqrt{p},\sqrt{p},24\sqrt{p})$-clustered.
Looking ahead, this allows us to design an algorithm that gives us an $\calO(\sqrt{p})$ approximation of the optimal cost.
In the general case, for an $(m,k,t)$-clustering, for any $t \ge 24\sqrt{p}$, we have an $\calO(m+k)$ approximation (by a similar analysis).

\smallskip{\em Double star.}
We define a double star as a graph of diameter $d=4$ with one central node, some middle nodes, and some leaf nodes.
The central and middle nodes create a star and the leaf nodes are connected to only one middle node, forming a cluster (see~\Cref{fig:doublestar} for an example with $6$ clusters).
In each cluster, the leaf nodes are connected to the middle node in a star topology.
In our setting, given our aforementioned assumption on $(\sqrt{p},\sqrt{p},24\sqrt{p})$-clustering, parties create a double star network topology during the network creation stage, where each cluster represents a group of nodes of size $\sqrt{p}$ that interact frequently with each other.

\begin{figure}[htb!]
    \centering
    \includegraphics[scale=0.6]{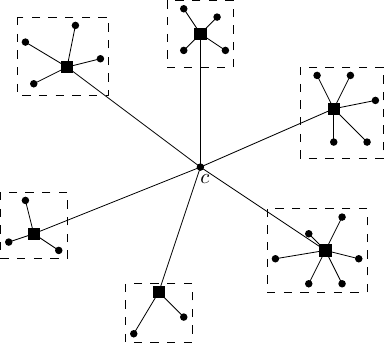}
    \caption{Double star graph with $6$ clusters in the dashed boxes. The node labelled $c$ is the centre node, and the middle nodes are the square nodes.}
    \label{fig:doublestar}
\end{figure}

\subsection{An $\calO(\sqrt{p})$ approximation algorithm}\label{sec:clusteralg}

\smallskip{\em Main ideas.}
We describe our algorithm that gives an $\calO(\sqrt{p})$ approximation of the optimal cost of the problem in~\Cref{alg:clusters}.
On a high level,~\Cref{alg:clusters} consists of two main ideas. 
The first idea is to again formulate and solve an LP on a double star to get the amount of capacities nodes should inject on their ends of their incident links. 
This LP is similar to the LP for a star described in~\Cref{lp:star} and we detail it in~\Cref{app:doubleStarLP}.
From~\Cref{thm:star}, we know that the capacity cost of the double star is at worst a $4$-approximation of the optimal capacity cost.

The second idea is to split the processing of transactions into processing transactions \emph{between} clusters and processing transactions \emph{within} clusters.
Here, we use the specific double star topology to ensure that we can apply~\Cref{alg:modif} (for processing transactions on a single star) to process transactions between clusters as well as within each cluster.
Finally, we use our clustering assumptions to allocate enough reserve capacity for each cluster to ensure that we can always shift the reserve capacities within a cluster from nodes with high to low reserves.
This ensures that each node always has sufficient reserve capacity to align the decisions on their transactions (whether to accept or reject) as closely as possible to the ideal decisions as output from the linear program.

\smallskip{\em Processing transactions between clusters.}
After running the linear program for the double star, let us rewrite each transaction in the sequence as $(s,t,x_i,y_i)$, where $s$ is the source node, $t$ is the target node, $x_i$ is the weight of the transaction, and $y_i$ denotes the fractional amount of the transaction accepted by the solution of the linear program.
For processing transactions between clusters, we look at the star created by the centre and middle nodes of the double star and treat these transactions as going only between these middle nodes. 
In~\Cref{alg:clusters}, this is handled by the function $\Call{BetweenClusters}{}$ which selects only transactions that go between clusters and treats clusters as a single node.
As this becomes processing transactions on a star, this allows us to use~\Cref{alg:modif} to process transactions between different clusters. 
In~\Cref{alg:clusters}, this is handled by the function $\Call{Solve}{X}$ in line~\ref{algline:solvebetween} in~\Cref{alg:clusters} which outputs the decisions using \Cref{alg:modif} for every pair of nodes.
Recall that \Cref{alg:modif} also gives us an upper bound on the \emph{reserve capacity} requirement for each cluster (middle node): for each middle node, we require reserves of size $\calO(\sqrt{p}M)$,
where $M$ is the maximum of the capacity of the channel between the middle node and the centre of the double star $c$ and the capacities of the channels between the middle node and the leaf nodes in its cluster.
We call these reserves \emph{cluster reserves}.

\smallskip{\em Processing transactions within clusters.}
Before processing transactions that go between nodes in a cluster, it is imperative that the cluster reserves are first shifted around within the cluster to ensure each leaf node in the cluster has enough reserves to process their between-cluster transactions (line~\ref{algline:solvebetween} in~\Cref{alg:clusters}).
This is handled by the function $\Call{BalanceInCluster}{}$ (described in~\Cref{alg:transport} in~\Cref{app:balanceincluster}), which modifies the solution of the double star LP such that cluster reserves are transferred from well-funded to depleted nodes.
%To do so, we note that every leaf node in a cluster is given $\calO(M)$ additional cluster reserves, which we denote by $R_{v}^{\mathcal{C\ell}}$ for node $v \in C_i$.
%Whenever the cluster reserve of a node depletes, we know from the clustering conditions that there are transactions of large volume inside the cluster.
%Then, the algorithm changes the solution of the LP so as to effectively simulate the transfer of these cluster reserves from well-funded to depleted nodes. This is handled by~\Cref{alg:transport}. 
%We depict the process of shifting cluster reserves when a transaction is accepted in~\Cref{fig:rcaccept}. 
Finally, after balancing the reserve capacities, \Cref{alg:clusters} treats each cluster as an independent star graph and uses~\Cref{alg:modif} to process transactions within each cluster (function $\Call{Solve}{}$ in line~\ref{algline:solve}). 
Using~\Cref{alg:modif} imposes a reserve capacity requirement of $\calO(\sqrt{p}M')$ times the capacity for every leaf node, where $M'$ is the capacity of the channel between the leaf node and the middle node.

\begin{comment}
The decisions between middle nodes when processing transactions between clusters create dependencies for the decisions on the transactions between the leaf nodes within a single cluster:
some transactions have to be accepted and some rejected.
This moves some reserves.
For every leaf node, we add there $\calO(M)$ additional cluster reserves (the middle node has $\calO(\sqrt{p}M)$ cluster reserves and every out of $\sqrt{p}$ nodes has $\calO(M)$ reserves).
The moving around changes the solution of the $LP$ slightly and incurs a cost of at most $\calO(\sqrt{p})$ times the rejection cost of the transactions between middle nodes, the details are explained in \Cref{sec:technical_balcing}.
In~\Cref{alg:clusters}, this is handled by the function
$\Call{BalanceInClusters}{X_t,S_c}$ which takes a sequence of transactions from $X_t$ that are not decided by $S_i$ and are in cluster $C$, and ensures that the set of decisions $S_i$ can be done.
Finally, after balancing the reserve capacities, \Cref{alg:clusters} treats each cluster as an independent star graph and uses~\Cref{alg:modif} to process transactions within each cluster (function $\Call{Solve}{}$ in line~\ref{algline:solve} of~\Cref{alg:clusters}). 
Using~\Cref{alg:modif} imposes a reserve capacity requirement of $\calO(\sqrt{p}M')$ times the capacity for every leaf node, where $M'$ is the capacity of the channel between the leaf node and the middle node.
\end{comment}

\begin{algorithm}[t]
  \begin{algorithmic}[1]
    \Require Sequence of transactions $X_t$ with $LP$ solution, clustering $\mathcal{C} = \{C_1,C_2,\dots\}$
    \Ensure Decisions to accept or reject transactions
    \State $X_i \gets \Call{BetweenClusters}{X_t,\mathcal{C}}$
    \State $S_i \gets \Call{Solve}{X_i}$ \label{algline:solvebetween}
    \For{$C \in \mathcal{C}$}
        \State $X_C \gets \Call{BalanceInCluster}{X_t, S_i, C}$
        \State $S_C \gets \Call{Solve}{X_C}$ \label{algline:solve}       
    \EndFor
    \State \Return $S_i \cup S_b \cup (\bigcup_{C\in \mathcal{C}}X_C)$
  \end{algorithmic}
  \caption{Approximation algorithm for clustered nodes.}\label{alg:clusters}
\end{algorithm}

\smallskip{\em Properties of balancing reserve capacities.}
We now prove several important properties of transporting reserve capacities within clusters as specified by $\Call{BalanceInClusters}{}$, ensuring the correctness of $\Call{BalanceInClusters}{}$. 
First,~\Cref{lem:whenMove} states that we can modify the solution of the linear program (i.e., the fractional amount of transactions accepted by the linear program) to move the reserves capacities from the source $s$ to target $t$, see \Cref{fig:rcaccept}.
From the statement of~\Cref{lem:whenMove}, we note that the only transactions that do not allow moving reserves are transactions that are fully accepted ($y_i = x_i$) or rejected ($y_i = 0$). Then, we quantify how much capacity we can move from one set of nodes to another in~\Cref{lem:howMuchMove}. Lastly, we show in~\Cref{lem:transport} that we can always transport some reserves between any two sets of nodes.
We defer the proofs of these lemmas to~\Cref{app:proofwhenmove,app:proofhowmuch,app:prooftransport}.

\begin{figure}[htb!]
    \centering
\includegraphics[scale=0.7]{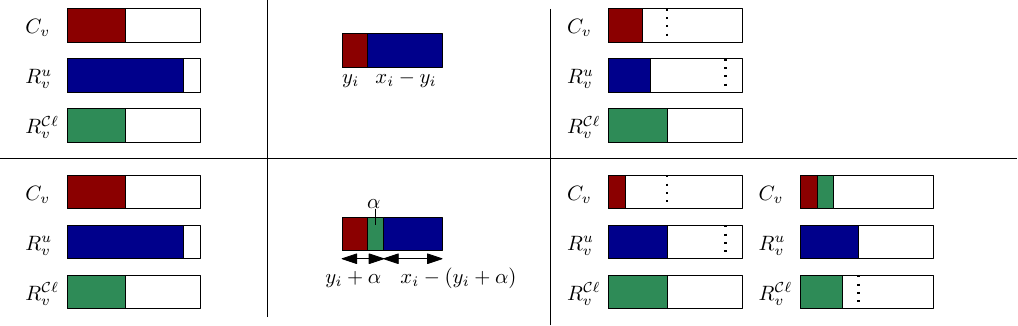}
    \caption{Movement of capacities, reserves and cluster reserves when accepting a transaction $x_i$ from $v$ to $u$ where $v,u \in C_i$ for some cluster $C_i$.
    The case on the top represents the typical movement of capacities and reserves when accepting $x_i$. The case on the bottom represents the treatment of the reserves (and cluster reserves) when the change in cluster reserves ($R_{v}^{\mathcal{C\ell}}$) is involved.
    The highlighted reserves are on $v$'s side of the channel between $v$ and the middle node. Similar changes are made on $u$'s side of the channel between $u$ and the middle node}
    \label{fig:rcaccept}
\end{figure}

\begin{restatable}[]{lemma}{whenmove}
\label{lem:whenMove}
    If there is a transaction $(s,t,x_i,y_i)$, we can change it to
    \begin{itemize}[leftmargin=0.4cm]
        \item $(s,t,x_i,y_i + \alpha)$ to simulate moving $\alpha$ reserves from $s$ to $t$, and
        \item $(s,t,x_i,y_i - \alpha)$ to simulate moving $\alpha$ reserves from $t$ to $s$.
    \end{itemize}
\end{restatable}

\begin{restatable}[]{lemma}{howmuch}
\label{lem:howMuchMove}
    Let $S$ be a set of nodes that is connected to $C_i\setminus S$ for some cluster $C_i$ by volume
    of at least $6|S|\cdot M$. Then at least $|S|M$ volume can be moved from $S$ to $C_i \setminus S$.
\end{restatable}

\begin{restatable}[]{lemma}{transport}
\label{lem:transport}
Let $X_t'$ be a subsequence of transactions where transactions incident to $C_i$ 
    have volume at least $tr \ge M$, then \Cref{alg:transport} can transport
    at least $tr$ from any set of nodes $S\subset C_i$ to any other set $T\subset C_i$.
    The LP solution is modified by volume at most $\calO(\sqrt{p}\cdot tr)$.
\end{restatable}

%\begin{lemma}\label{lem:transport}
%    Let $X_t'$ be a subsequence of transactions where transactions incident to $C_i$ 
%    have volume at least $tr \ge M$, then \Cref{alg:transport} can transport
%    at least $tr$ from any set of nodes $S$ to any set $T$.
%    The $LP$ solution is modified by volume at most $\calO(\sqrt{p})$.
%\end{lemma}

We now have the requisite ingredients to state and prove (in~\Cref{app:proofmain}) the main theorem of this section.

\begin{restatable}[]{theorem}{main}
\label{thm:main}
    \Cref{alg:clusters} is $\calO(\sqrt{p})$ approximation of the problem where given a transaction sequence, the nodes are $(\sqrt{p},\sqrt{p}, 24\sqrt{p})$-clustered.
\end{restatable}

We also get the following corollary as a direct consequence of the proof of~\Cref{thm:main}.
\begin{corollary}
    Given $(m,k,t)$-clustered transaction sequence where $t \ge 24\sqrt{p}$,
    \Cref{alg:clusters} is $\calO(\max(m,k))$-approximation of the problem.
\end{corollary}

Finally, we end this section with two observations. The first observation effectively states that our algorithm just needs to focus on the total volume of partially accepted transactions and can ignore the volume of fully accepted and rejected transactions.
The second observation allows us to further reduce the approximation ratio if we assume clustering of a larger depth.

\smallskip{\em Practical relaxations.}
We observe that the proof of~\Cref{thm:main} requires a condition on the partially accepted transactions
(these are transactions which satisfy $0 < y_i < x_i$).
However, we stress that the volume of fully accepted or rejected transactions can be arbitrary.
Moreover, in every cluster, the middle node can interact with other clusters arbitrarily and this also does not affect the proof of~\Cref{thm:main}.

\smallskip{\em Hierarchical clustering.}
We can use a similar idea to further reduce the approximation ratio if we assume that within each cluster the clusters themselves are also clustered.
Namely, we first treat the top-level clusters as nodes and solve the problem as outlined by \Cref{alg:modif}.
Then, we have some transactions pre-accepted or pre-rejected and we now shift the focus to the next level clusters.
If these next-level clusters consist of other clusters, we treat them again as nodes.
We do this recursively until we reach the lowest-level clusters and then use~\Cref{alg:clusters} to solve the problem within the lowest-level clusters.
Note that we require the strength of the clustering to be bigger in order to achieve a similar approximation ratio in this setting.
Nevertheless, assuming the nodes satisfy sufficiently strong clustering conditions, we can achieve an approximation ratio of $\calO(p^{\frac{1}{3}})$ for two nested clusterings, and, in general, $\calO(p^{\frac{1}{k+1}})$ for $k$ nested clusterings (if the nodes in clusters are well distributed).
We stress that although hierarchical clustering allows us to achieve a stronger approximation factor, it also comes with additional structural assumptions. However, we believe there are strong use cases of hierarchical clustering, for instance using a payment channel network to handle payments among users in a large company. The clusters would naturally follow the clustering of users into divisions and sub-divisions, with the head of each division being the centre of the star in their corresponding cluster.

\section{Case study: the Lightning Network}\label{sec:eval}

To justify our transactions distribution assumption in~\Cref{sec:algcluster}, we perform an empirical case study of the largest and most commonly used PCN: Bitcoin's Lightning Network. Specifically, the research question we want to investigate is the following: 

\smallskip
\emph{Are transactions in the Lightning Network clustered according to the clustering conditions stated in~\Cref{sec:algcluster}?}
\smallskip

To address this question, we obtained the latest snapshot (dated September 2023) of the Lightning Network from the Lightning Network gossip repository~\cite{lngossip}. 
We further restrict our examination to the largest connected component which contains $20,348$ users and $310,657$ channels.
\Cref{fig:degree_histogram} shows the distribution of degrees in the largest connected component, where a large number of nodes ($\sim 10\%$) are of degree $\leq 2$.

\smallskip{\em Estimating the source and target of transactions.}
The main challenge in extracting transactions data from the Lightning Network is that the Lightning Network, by design, only provides information about the costs of transactions across channels, but does not provide information about transactions between users.
Nevertheless, it is reasonable to assume that transactions between users are related to the channel capacities that connect them. 
Furthermore, we are mainly concerned with finding out if the clustering conditions hold for transactions between ordinary users with low degrees rather than between, e.g., large payment hubs with high degrees.
We hence propose the following methodology: we first remove high degree nodes (i.e., nodes with degree $\geq 3$) in the network and connect their neighbours to each other. Let us denote this resulting reduced network by $G$. In $G$, we assume nodes that are directly connected to each other transact with each other.
Our methodology of generating the reduced network $G$ is detailed in~\Cref{alg:users-network}.

\begin{algorithm}[htb!]
    \caption{Construct New Graph \( G \) from Graph \( H \)}
    \begin{algorithmic}[1]
        \Require Graph \( H = (V, E) \), $cost_H : E \rightarrow \mathbb{R}_{\geq 0}$ 
        \Ensure Graph $G$ of users and $vol$ assigning volume to its edges.
        \State Find the largest strongly connected component \( C \subseteq H \)
        \State \( users \gets \{v \in C \mid 1 \leq \deg(v) \leq 2\} \)
        \State  \( hubs \gets V \setminus users \)
        \State \( E_{users} \gets \{ \{u, v\} \mid u, v \in users \text{ and } N[u] \cap N[v] \cap hubs \neq \emptyset \} \cup E(G[users]) \)
        \State \( G \gets (users, E_{users}) \)
        \State \( cost_G(u, v) \gets cost_H(u, v) \) for all $\{u, v\} \in E_{users}$.
        \State \( cost_G(u, v) \gets \infty \) for all $u, v \in users$ and $\{u, v\} \notin E_{users}$.
        \ForAll{$v \in hubs$}
            \For{$s_1, s_2 \in users \cap N(v)$}
                \State $cost_G(s_1, s_2) = \min(cost_G(s_1, s_2), cost_H(s_1, v) + cost_H(v, s_2))$
            \EndFor
        \EndFor
        \State \( vol(u, v) \gets 1 / (cost_G(u, v) + 1) \) for all $u, v \in S$.
        \State \Return \( G \) and $vol$
    \end{algorithmic}
  \caption{Algorithm for estimating the communication network of regular users.}\label{alg:users-network}
\end{algorithm}

\smallskip{\em Estimating the volume of transactions.}
Recall that a crucial clustering condition is that the volume of transactions that stay within a cluster is larger than the volume of the transactions that go between clusters.
However, the privacy principles behind the design of the Lightning Network renders such transaction volume information hidden.
Thus, a second challenge in our setting is extracting estimates of the size of transactions between users. 
%In order to estimate this figure, we simply use the channel capacities between any two users in $G$ as an approximation of the average transaction size for transactions in going in both directions along the channel. 
In order to estimate transaction volume, we assume that the volume of transactions between two connected users in $G$ is inversely proportional to the cost of a transaction.
Namely, that the volume of transactions between users $i$ and $j$ is $1 / (c_{i,j} + 1)$ where $c_{i,j}$ is the minimum cost of sending a transactions of 100 milli-satoshi from $i$ to $j$.

Note that the volume of transactions is estimated under the simplifying assumptions that all transactions have the same fixed volume, and therefore the cost for a transaction between a pair of users is fixed and specified using $cost_H$ on input.

We start with the largest strongly connected component of $H$ as it is the largest set of users such that any user can send a transaction to any other.
Also, to simplify the experiment, we assume that the volume of transactions from a user $u$ to another user $v$ is the same as the volume of transactions in the reverse order (from $v$ to $u$).

This approach of computing costs for each pair of users runs in time linear with the number of edges of $G$. Note that the number of updates of a cost is at most the number of edges adjacent to the set $users$, which is bounded by $O(|users|)$ as the degree of each vertex is bounded.

\smallskip{\em Results.}
\Cref{fig:network-clustering} and~\Cref{tab:clustering-details} shows that we obtain 18 clusters with average within-cluster transaction volume of 826.1, and between-cluster transaction volume of 4.27. This gives us a ratio of 193.2 of within-cluster and between-cluster transactions, which confirms our experimental hypothesis.

\begin{table}[htb!]
\centering
\caption{Details of individual clusters.}
\label{tab:clustering-details}
\begin{tabular}{lrrrl}
\toprule
{Cluster} & Cluster Size &  Inside Volume &  Between Volume &   Ratio \\
\midrule
1  &           473 &               47.85 &                  2.73 &   17.53 \\
2  &           309 &                4.32 &                  0.33 &   13.09 \\
3  &           136 &                2.97 &                  0.66 &     4.5 \\
4  &           126 &               33.32 &                  0.29 &   114.9 \\
5  &           123 &                2.18 &                  1.67 &    1.31 \\
6  &           118 &                3.52 &                  0.29 &   12.14 \\
7  &            96 &              715.15 &                  0.33 &  2167.1 \\
8  &            76 &                0.79 &                  0.99 &     0.8 \\
9  &            67 &                0.18 &                  0.12 &     1.5 \\
10 &            64 &                0.12 &                  0.07 &     1.7 \\
11 &            59 &                4.20 &                  0.00 &       $\infty$ \\
12 &            57 &                0.37 &                  0.13 &     2.8 \\
13 &            54 &                0.65 &                  0.00 &       $\infty$ \\
14 &            54 &                3.09 &                  0.71 &    4.35 \\
15 &            31 &                1.06 &                  0.04 &    26.5 \\
16 &            30 &                0.22 &                  0.11 &     2.0 \\
17 &            25 &                6.07 &                  0.04 &   151.8 \\
18 &            19 &                0.04 &                  0.04 &     1.0 \\
\bottomrule
\end{tabular}
\end{table}

\begin{figure}[!htb]
  \centering
  \begin{minipage}[b]{0.4\linewidth}
    \includegraphics[width=\linewidth]{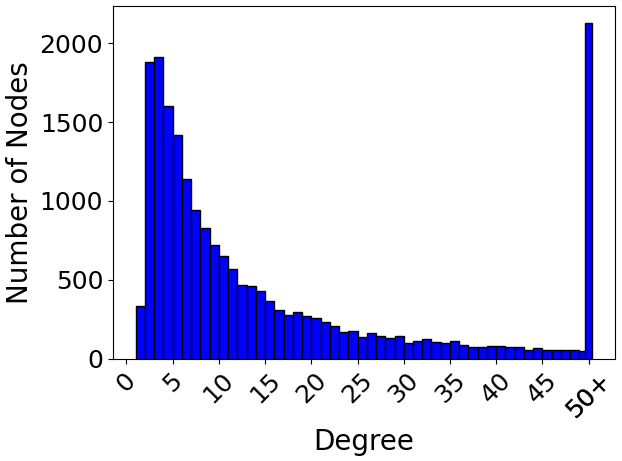}
    \caption{Degree distribution of the nodes in the largest connected component of the Lightning Network.}
    \label{fig:degree_histogram}
  \end{minipage}
  \hfill
  \begin{minipage}[b]{0.4\linewidth}
    \includegraphics[width=\linewidth]{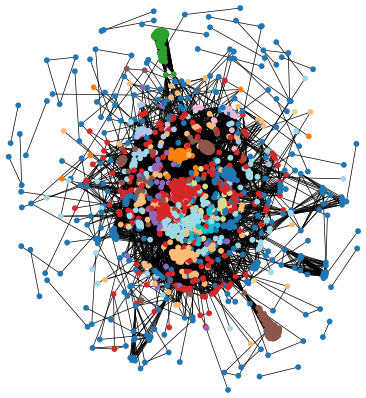}
    \caption{Clustering of the nodes in the reduced network $G$. Each color represents a different cluster.}
    \label{fig:network-clustering}
  \end{minipage}
\end{figure}

\section{Conclusion}\label{sec:conclusion}

In this work, we generalise a model from processing transactions over a single payment channel to an entire PCN. To this end, we first define a network creation stage of the problem which examines how to create a network with minimal capacity cost, and we show that the star is a $2$-approximation of the optimal capacity cost.
We then propose two algorithms to address weighted transaction processing over a network. The first algorithm assumes a star topology over the $p$ parties, but apart from that works without any further assumptions on the transaction sequence, and gives an approximation ratio of $\calO(p)$. The second algorithm assumes a double star topology and a clustering of the transaction sequence, and gives an approximation ratio of $\calO(\sqrt{p})$. 
Finally, we also perform an empirical case study estimating transaction information from the topology of the Lightning Network, and our estimates suggest that our clustering assumptions regarding the input transaction sequence are close to reality.

We believe our work is an important first step in the direction of designing optimal networks and algorithms for transaction processing in PCN.
We highlight two interesting directions for future work.
First, it would be interesting to extend our model and algorithms to the online setting, which will also complement the analysis of~\cite{AvarikiotiB0W19,online}.
%Second, extending our model of actions to include options like shifting capacities in a cycle in order to refund depleted channels, which also corresponds to rebalancing~\cite{AvarikiotiPSSTY21,SchmidSY23} in PCNs is an important open question.
Secondly, developing better estimates regarding transaction information in real world PCNs would also be paramount to developing better optimisation algorithms and heuristics.

\bibliography{refs}

\appendix

\section{Omitted proofs}
\subsection{Proof of~\Cref{lemma:lp}}\label{app:prooflp}
\lp*
\begin{proof}
  We first observe that adding channels to a graph can only decrease the capacity cost (adding channels can only make the shortest paths shorter).
  This means that the capacity cost is the lowest for the complete graph, no matter what is the optimal graph topology (as computed by minimising the \emph{total} cost of both capacity and channel creation).
  The LP computes the optimal capacity for the complete graph $K_p$, which is a superset of every graph, including the optimal graph.
\end{proof}

\subsection{Proof of~\Cref{thm:star}}\label{app:proofstar}
\star*
\begin{proof}
  Since for the complete graph $K_n$, we have $AC(\pseq, K_n) \le AC(\pseq, G)$, we simply need to show that $AC(\pseq, G') \le d \cdot AC(\pseq, G)$.
  We modify the solution of the LP for the complete graph by adding at most $d$-times the capacity and prove that this graph processes all transactions.

  We start with graph $G'$ with zero capacities.
  For all pairs of nodes $v$ and $u$, let the capacity between them be $w_{v,u}$.
  We find the shortest path between them in $G'$, since the diameter is at most $d$, the length of the path is at most $d$.
  Then, we increase the capacity of every channel on the path by $w_{v,u}$.
  We denote this capacity to a special bucket that will be used to process transactions that the LP sends along $u,v$.
  At the start at edge $k,l$, where $k$ is closer to $v$ and $l$ is to $u$, we increase the balance at the $k$'s side
  by $C_{v,u}^0$ (and $l$'s balance is $C_{u,v}^0$), the balance at the edge $v,u$ at the start as given by the LP.

  After we increase the capacities of $G'$ for every pair of nodes, we process the transactions.
  Any time the transaction of size $s$ was forwarded along some edge $u,v$, we can transfer $s$ on the shortest path
  between $u$ and $v$ from the bucket on the edge denoted $u,v$.
  If the bucket for some edge becomes negative, the solution of LP is infeasible.

  By this, we can process all transactions.
  Since the diameter is at most $d$, we know that the increase in capacity is at most $d$-times.
\end{proof}

\subsection{Proof of~\Cref{thm:pairwise}}\label{app:proofpairwise}
\pairwise*

\begin{proof}
    Every node has $p-1$ connections while the size of the channel is viewed as the minimum of the two sizes as denoted by the LP.
    The reserves are $\sqrt{3}$ times that value.
    Moreover, every channel in a star has funds amount as dictated by the LP.
    Altogether it gives $1 + (p-1) \cdot \sqrt{3}$ times of $\mathcal{C}$.

    The algorithm ensures that the in every channel on the right side, there is at least the amount prescribed by the LP.
    This means that all fully accepted ($x_i = y_i$) transactions are accepted.
    Moreover, from \cite{SchmidSY23}, we know that the rejection cost between two nodes is at most
    $\sqrt{3} + 1$ times the rejection cost of the optimal rejection cost $\mathcal{R}$. Thus,~\Cref{alg:modif} incurs a rejection cost of at most $(\sqrt{3} + 1)\mathcal{R}$ and capacity cost of at most $(1 + (p-1)\cdot \sqrt{3})\mathcal{C}$.
\end{proof}

\subsection{Proof of~\Cref{lem:whenMove}}\label{app:proofwhenmove}
\whenmove*
\begin{proof}
    Every node has $p-1$ connections while the size of the channel is viewed as the minimum of the two sizes as denoted by the LP.
    The reserves are $\sqrt{3}$ times that value.
    Moreover, every channel in a star has funds amount as dictated by the LP.
    Altogether it gives $1 + (p-1) \cdot \sqrt{3}$ times of $\mathcal{C}$.

    The algorithm ensures that the in every channel on the right side, there is at least the amount prescribed by the LP.
    This means that all fully accepted ($x_i = y_i$) transactions are accepted.
    Moreover, from \cite{SchmidSY23}, we know that the rejection cost between two nodes is at most
    $\sqrt{3} + 1$ times the rejection cost of the optimal rejection cost $\mathcal{R}$. Thus,~\Cref{alg:modif} incurs a rejection cost of at most $(\sqrt{3} + 1)\mathcal{R}$ and capacity cost of at most $(1 + (p-1)\cdot \sqrt{3})\mathcal{C}$.
\end{proof}

\subsection{Proof of~\Cref{lem:howMuchMove}}\label{app:proofhowmuch}
\howmuch*
\begin{proof}
    From the condition $3$ of the clustering conditions, we know that at least $2|S|M$ volume moves from the cluster
    and $2|S|M$ volume to the cluster.
    %We have no guarantee about how accepted are the transactions.
    For the sake of contradiction, suppose that we can move at most $m < |S|M$ from $S$ to $C_i \setminus S$.
    
    From \Cref{lem:whenMove}, we know that we cannot move only on fully accepted transactions from $S$ and rejected transactions from $C_i\setminus S$.
    Let $m'$ be the incoming volume to $S$ that can be moved.
    Then in the subsequence in the solution of $LP$, $S$ receives $m'$ volume
    and gives up $2|S|M - (m-m')$ volume.
    That means on average it lost at least $2|S|M - (m-m')-m' = 2|S|M - m > |S|M$.
    From the Dirichlet principle, at least one node lost more than $M$.
    This is impossible since $M$ is the highest capacity in $C_i$ and the solution of the $LP$ ensures that for any subsequence the difference is at most $M$.
\end{proof}

\subsection{Proof of~\Cref{lem:transport}}\label{app:prooftransport}
\transport*

\begin{proof}
    Since there was $tr$ transported, inside the cluster there was at least $24\cdot tr\cdot |C_i|$ transported,
    and next to $S$, there was at least $6tr|S|$ transported.
    From \Cref{lem:howMuchMove}, that means that at least $tr|S|$ goes away from $S$ to other nodes.

    Let $S_0 = S$ and $S_i$ be the nodes that we can transport nonzero amount from $S_{i-1}$.
    Similarly, let $T_0 = T$ and $T_i$ be the nodes that we can transport nonzero amount to $T_{i-1}$.
    Since the clustering conditions, we have that the size of $S_i$ and $T_i$ is increasing if their size is at most half of the cluster.
    That means they intersect at some point.

    We look at the minimum $mn$ that can be transported and we change the transactions such that in $v$ the reserves are decreased by $mn$ and in $u$ they are increased by $mn$.
    After the change, we know that the clustering conditions still hold,
    now we just need to decrease the inequalities by $mn$.
    Therefore, we can find the path again.
\end{proof}

\subsection{Proof of~\Cref{thm:main}}\label{app:proofmain}
\main*
\begin{proof}
    We are using a double star, which means it is at least a $4$ approximation of the optimal topology
    from \Cref{thm:star}. 
    
    The algorithm uses \Cref{alg:modif} for every pair of clusters (between middle nodes)
    to accept and reject transactions between them.
    This incurs cost for additional capacity equal $\calO(\sqrt{p})$ times the cost of the capacity cost of the $LP$
    for the capacity between clusters.
    The rejection cost is at most $(\sqrt{3}+1)$ times higher than the rejection cost between clusters.
    All of this follows from \Cref{thm:pairwise}.

    Next, from \Cref{lem:transport}, we know that we can transport reserves such that everything between clusters
    is accepted or rejected from the previous algorithm.
    It changes the values of the $LP$ by at most $\calO(\sqrt{p})$ times the volume of rejected transactions between clusters.

    Finally, again from \Cref{thm:pairwise}, we can solve the problem inside the clusters.
    This incurs capacity cost $\calO(\sqrt{p})$ times the capacity cost of the solution inside the clusters.
    Moreover, the cost for rejection is $(\sqrt{3} + 1)$ times the rejection cost in the modified solution of $LP$ ($X_t'$).
    The modified rejection cost of solution of $LP$ is the rejection cost of the original solution plus
    at most $\calO(\sqrt{p})$ times the rejection cost of transactions between clusters (from \Cref{lem:transport}). Altogether, we get that the algorithm is $\calO(\sqrt{p})$ approximation algorithm.
\end{proof}

\section{Omitted details of star topology solution}

\subsection{Modified channel algorithm details}\label{app:channelalg}
We now describe the decision making process in~\Cref{alg:modif} for transactions going from $u$ to $v$ (the other direction is symmetric).
%We describe the decisions process only for transactions going from $u$ to $v$ as the process for transactions going in the other direction is symmetric.
~\Cref{alg:modif} accepts a transaction as long as the reserve $R_u$ is above a threshold (i.e., $R_u \ge \frac{\sqrt{3}-1}{2}M$, line~\ref{algline:balstart}). Accepting a transaction decreases the capacities $C_u$ and $R_u^v$ of $u$ by $y_i$ and $x_i-y_i$ respectively (refer to the ``accept'' case in~\Cref{fig:channeltostar}).
If accepting the transaction would make $R_u$ decrease by too much and the LP demands to accept the transaction \emph{weakly} (i.e., $\frac{y_i}{x_i} < \frac{\sqrt{3}}{\sqrt{3}+1}$, line~\ref{algline:rej}), the algorithm rejects the transaction. Rejecting a transaction from say $u$ to $v$ decreases both $C_u$ and $R_u^v$ of by $y_i$ (refer to the ``reject'' case in~\Cref{fig:channeltostar}).
If the reserve on $u$'s side after accepting the transaction would be too small (i.e., $R_u < \frac{\sqrt{3}-1}{2}M$, line~\ref{algline:forward}), the algorithm looks at all subsequent transactions \emph{between} $u$ and $v$ in the transaction sequence. 
It accepts all transactions that increases $R_u$ while rejecting weakly accepted transactions from $u$ to $v$ (rejecting them also increases $R_u$).
From \emph{strongly} accepted transactions (i.e., $\frac{y_i}{x_i} \ge \frac{\sqrt{3}}{\sqrt{3}+1}$) from $u$ to $v$, the algorithm tries to accept
all of them while simulating the movement of capacities between the reserves until either $R_u > \frac{\sqrt{3}-1}{2}M$ or $R_u < 0$.
If $R_u > \frac{\sqrt{3}-1}{2}M$, then all transactions can be accepted.
If $R_u < 0$, the algorithm rejects strongly accepted transactions until $R_u > \frac{\sqrt{3}-1}{2}M$.
The proof in~\cite{SchmidSY23} shows that doing so maintains the approximation ratio at $\sqrt{3}+1$ for the rejection cost.
%For readability purposes, we defer the description of the technical details of~\Cref{alg:modif} to~\Cref{app:channelalg}.

\smallskip{\em Algorithm details.}
The simulating of the reserves is implemented by a heap.
The heap is empty while $R_u \ge \frac{\sqrt{3}-1}{2}M$, and
when $R_u < \frac{\sqrt{3}-1}{2}M$ it starts to fill up.
When $R_u < 0$, then the decision is changed for transactions from the top of the heap until $R_u \ge \frac{\sqrt{3}-1}{2}M$ again.

Observe that for \Cref{alg:modif} the funds on both sides can shift and the reserves are not affected.
This allows us to compose the algorithm from pairwise reserves with shared funds denoted by the linear program.

The algorithm describes decisions for the star, so a transaction goes from $u$ to $c$ (star center) and then from $c$ to $v$.
That means it needs to traverse two channels.
The algorithm sees it as only pairwise interaction.
Since the decisions for the two channels are always the same for both channels (a transaction is either accepted or rejected for both), we can have the reserves stored twice, once between $u$ and $c$, the second time between $v$ and $c$.
The value $R_u$ is then the reserves at $u$ in channel $u,c$ and at $c$ in the channel $c,v$.
Both of these reserves move in sync.

\section{Omitted details of clustering algorithm}
\subsection{Stochastic Block Model details}\label{app:sbm}
The planted partition model creates a random graph where nodes are partitioned into clusters.
Inside a cluster, two nodes are connected with probability $q_1$, otherwise, they are connected with probability $q_2$.
The planted partition model creates an unweighted and undirected graph.
Our graph created by the packets is directed, weighted, and temporal.
These changes call for modification of the planted partition model.

The biggest change to the model is brought by the temporality.
Some edges follow after others, we require that all of the conditions hold for
some continuous subsequence of edges incident to some cluster of bigger size.

To make the graph weighted, we view $q_1$ and $q_2$ as ratios between volumes of the packets between the clusters and inside the cluster.
This view is expressed by the condition $1$ of the clustering conditions.

To make the graph directed, we allow for arbitrary orientation of the edges with the condition that at least $\frac{1}{3}$ of the volume needs to be directed in both directions.
It is expressed by condition $2$ of the clustering conditions.
Note that $\frac{1}{3}$ of the volume can be oriented in any direction, no matter the other packets.

Finally, the condition $2$ of the clustering conditions ensures that some volume is incident to every node.

Note that all of the conditions would be satisfied with high probability in the stochastic block model.
Moreover, there are other conditions that would be satisfied (for instance any two clusters would send packets between each other at the same rate) that we do not require.
From the above the graph generated by the planted partition model with $\sqrt{p}$ partitions of size $\sqrt{p}$ and $q_2 = \frac{1}{p}$ and $q_1 = 1-q_2$ satisfies all clustering conditions with high probability.

\subsection{Linear program for double star}\label{app:doubleStarLP}
Let the set of leaf nodes be $L$ and the set of middle nodes be $M$, and $c$ denote the central node of the star.
For a leaf node $u$, let $M(u)$ be the middle node corresponding to $u$'s cluster.

For a leaf node $v$ and natural number $i$, we use
$C_{v,i}$ and $C_{v,i}'$ to denote the capacity at the channel $(v,M(v))$ after processing the $i$-th transaction on the side of $v$ and $M(v)$ respectively.
Moreover let $C_v$ denote the total capacity of the channel $(v,M(v))$, i.e., $C_v = C_{v,i} + C_{v,i}'$.
For a middle node $v$ and $i$, we use
$M_{v,i}$ and $M_{v,i}'$ to denote the capacity at the channel $(v,c)$ after processing the $i$-th transaction on the side of $v$ and $c$ respectively.
Moreover, let $M_v$ denote the total capacity of the channel $(v,c)$.

\begin{align}\label{lp:doublestar}
  \qquad\qquad \text{minimise }&\sum_{u \in L} C_u + \sum_{u \in M} M_u + \\
  &\sum_{i} f(x_i-y_i) + \frac{m(x_i-y_i)}{x_i}\\
  \forall i :& 0 \le y_i \le x_i \notag\\
  \forall (v,u,x_i): &  C_{v}, C_{u} \ge x_i\notag\\
                     &  C_{v,i} = C_{v,i-1} - y_i\notag\\
                     &  C_{v,i}' = C_{v,i-1}' + y_i\notag\\
                     &  C_{u,i}' = C_{u,i-1}' - y_i\notag\\
                     &  C_{u,i} = C_{u,i-1} + y_i\notag\\
  \forall (v,u,x_i) \text{ between clusters: }&  M_{M(u)}, M_{M(v)} \ge x_i\notag\\
                     &  M_{M(v),i} = M_{M(v),i-1} - y_i\notag\\
                     &  M_{M(v),i}' = M_{M(v),i-1}' + y_i\notag\\
                     &  M_{M(u),i}' = M_{M(u),i-1}' - y_i\notag\\
                     &  M_{M(u),i} = M_{M(u),i-1} + y_i\notag
\end{align}

\subsection{Algorithm to transport cluster reserves}\label{app:balanceincluster}

Here we detail~\Cref{alg:transport}, which transfers cluster reserves from well-funded to depleted nodes \emph{within} a cluster, so that the cluster nodes have sufficient reserves to handle transactions that go \emph{between} the clusters.

\begin{algorithm}[t]
  \begin{algorithmic}[1]
    \Require Subsequence of transactions $X_t$ inside one cluster, two nodes $v$, $u$, amount $A$.
    \Ensure Modified $X_t$ such that $A$ in reserves is transported from $v$ to $u$.
    \Procedure{Transport}{$X_t$, $v$, $u$, $A$}
        \State $M \gets \Call{PathBetween}{v,u}$
        \State $mn \gets \Call{MinimumOnPath}{M}$ 
        \State $mn \gets \min(mn, A)$
        \For{$(dir, x,y,x_i,y_i)\in M$}
            \If{$dir$ from $v$ to $u$}
                \State $y_i \gets y_i - mn$
            \Else
                \State $y_i \gets y_i + mn$
            \EndIf
        \EndFor
        \If{$A-mn > 0$}
            \State $\Call{Transport}{X_t, v,u,A-mn}$ 
        \EndIf
    \EndProcedure
  \end{algorithmic}
  \caption{Transporting reserves between nodes.}\label{alg:transport}
\end{algorithm}

\end{document}